\documentclass[12pt]{article}
\usepackage[margin=1in]{geometry}
\usepackage{setspace}
\usepackage{algorithm}
\usepackage{algpseudocode}
\usepackage{tikz}
\usepackage{enumitem}
\usepackage{caption} 
\usepackage{endnotes}

\usepackage{amssymb}
\usepackage{amsmath}
\usepackage{siunitx}
\usepackage{booktabs}
\usepackage{amsthm}
\usepackage{hyperref}
\usepackage{cleveref}
\hypersetup{colorlinks=true,citecolor=black,linkcolor=black,urlcolor=black,breaklinks=true}
\usepackage{subcaption}
\usepackage{bm}
\usepackage{eurosym}
\usepackage{pgf-pie}
\usepackage{xcolor}
\usepackage{svg}
\usepackage{amsthm}

\usepackage{natbib}
\bibpunct[, ]{(}{)}{,}{a}{}{,}%
\usepackage{etoolbox}

\providecommand{\keywords}[1]{\small	\textbf{\textit{Keywords---}} #1}

\newtheorem{corollary}{Corollary}
\newtheorem{theorem}{Theorem}
\newtheorem*{theorem*}{Theorem}
\newtheorem*{proposition*}{Proposition}

\theoremstyle{definition}
\newtheorem{example}{Example}
\newtheorem{definition}{Definition}
\newtheorem{assumption}{Assumption}
\newtheorem{remark}{Remark}
\newtheorem{criterion}{Criterion}
\newtheorem{fact}{Fact}

\DeclareMathOperator*{\argmax}{arg\,max}

\newcommand{\normtwo}[1]{\left\lVert #1 \right\rVert_2}

\newcommand*\diff{\mathop{}\!\mathrm{d}}

\title{On the Design of Stochastic Electricity Auctions}

\author{Thomas Hübner\thanks{Thomas Hübner (thuebner@ethz.ch) is with the Department of Information Technology and Electrical Engineering, ETH Zürich. This work was supported by the Swiss Federal Office of Energy’s SWEET program under the PATHFNDR project (Grant No. SI/502259). The majority of this work was conducted while the author was a visiting researcher at the Department of Economics and Nuffield College, University of Oxford, hosted by Elizabeth Baldwin and Paul Klemperer.}}

\date{}

\begin{document}

\maketitle

\vspace{-0.5cm}

\begin{abstract}
Electricity is typically traded in day-ahead auctions because many power system decisions, such as unit commitment, must be made in advance. However, when wind and solar generators sell power one day ahead, they face uncertainty about their actual production. In current day-ahead auctions, this uncertainty cannot be directly communicated, leading to inefficient use of renewable energy and suboptimal system decisions.
We show how this problem can be addressed using the concept of equilibrium under uncertainty from microeconomic theory. In particular, we demonstrate that electricity contracts should be conditioned not only on the time and location of delivery, but also on the state of the world (e.g., whether it will be windy or calm).
This requires a precise definition of the state of the world. Since there are infinitely many possible definitions, criteria are needed to select among them. We develop such criteria and show that the resulting states correspond to solutions of an optimal partitioning problem. Finally, we illustrate how these states can be computed and interpreted using a case study of offshore wind farms in the European North Sea.
\end{abstract}


\section{Introduction}

Day-ahead auctions are the central platform for electricity trading in most markets worldwide~\citep{Cramton2017}.
These auctions typically take place around noon and enable the trading of contracts for power delivery on the following day.

Trading electricity in advance of delivery is necessary because operational decisions must be made beforehand. For example, a thermal power plant must decide whether to start up before delivery. To avoid the risk of producing at a loss, it would therefore like to sell its output before starting up. Similarly, a manufacturer may need to set its production schedule several hours before production. To avoid the risk that prices exceed the value of production, it would therefore like to buy electricity before fixing this schedule.

In contrast, firms with wind and solar generation prefer to sell electricity as close to the time of delivery as possible, as they face uncertainty about their actual output.
If renewable producers are required to sell day-ahead because most buyers prefer to transact well before real time, they face two imperfect options. They can sell too little and later attempt to sell any remaining production in intraday or balancing markets, if a buyer can be found, and otherwise curtail generation. Alternatively, they can sell too much and risk having to repurchase electricity later, potentially at very high prices. In either case, renewable generation is utilized less efficiently than it would be if trading occurred closer to delivery.

In this article, we argue that this inefficiency stems from the \textit{incompleteness} of the day-ahead auction.
Classic results in microeconomic theory show that when commodities are traded in advance under uncertainty, efficiency requires that delivery depends not only on the time and location of delivery but also on the \textit{state of the world}~\citep{debreu1959theory,arrow1964theory,radner1982equilibrium,mas1995microeconomic}.
However, in current day-ahead auctions, electricity delivery is contingent on time (e.g., 8 am or 10 am tomorrow) and location (e.g., eastern or western Denmark), but not on the state of the world (e.g., whether it is windy or calm).

We are not the first to argue that increased uncertainty from renewable generation calls for changes to the day-ahead auction design.
A large body of work has addressed this issue, including \cite{wong2007pricing,pritchard2010single,morales2012pricing,morales2014electricity,zavala2017stochastic,kazempour2018stochastic,bjorndal2018stochastic,zakeri2019pricing,exizidis2019incentive,mays2021quasi,ratha2023moving,mays2024sequential,dvorkin2025regression,singhal2026truthful}. 
In this literature, the problem is typically framed as follows: current day-ahead auctions rely on a deterministic market-clearing optimization model, whereas uncertainty from renewable generation suggests that a stochastic formulation would be more appropriate. The literature then examines the implications of adopting stochastic market-clearing models for auction design. A common finding is that these ``stochastic'' auction designs exhibit various shortcomings, often leading to the conclusion that they may be impractical.

In contrast to the articles above, we approach the problem from a different perspective by relating it to the literature on \textit{equilibrium under uncertainty} and \textit{incomplete/complete markets}~\citep{debreu1959theory,arrow1964theory,radner1982equilibrium,mas1995microeconomic}.\endnote{The theory of equilibrium under uncertainty is often applied to long-term electricity markets, where hedging instruments (e.g., contracts for difference) are known to be incomplete: agents cannot hedge against all states of the world. This incompleteness can distort investment decisions and may require government intervention (see, e.g.,~\cite{abada2026market}). We argue that electricity markets are also substantially incomplete in the \textit{short term}, due to the large share of renewable generation.}
As noted above, this literature addresses essentially the same issue (trade under uncertainty) and offers a conceptually simple solution: state-contingent electricity delivery combined with ex ante payments.
When this is done, the classic results for Walrasian (competitive) equilibrium in deterministic settings extend directly to environments with uncertainty. In particular, the desirable properties of current ``deterministic'' day-ahead auctions -- such as strategy-proofness in the large, budget balance, individual rationality, and efficiency -- naturally carry over to a corresponding ``stochastic'' auction design. The resulting auction design therefore provides a simple extension of current day-ahead auctions: it incorporates uncertainty while preserving their key properties and avoiding the deficiencies associated with stochastic market-clearing approaches in the existing literature.

The central question for day-ahead auction design -- one that is not addressed in the equilibrium under uncertainty literature -- is how an auctioneer can obtain a precise definition of what constitutes a ``state of the world''.
The fundamental problem in defining a state is that there are infinitely many possible candidates, while only finitely many can be selected.\endnote{The need for discretization in (electricity) markets arises primarily from computational considerations. Time, which lies in the continuous interval $[0,24]$, is typically discretized into 24 hourly or 96 quarter-hourly intervals. Similarly, continuous sources of uncertainty, such as wind speed in $\mathbb{R}$, must be discretized into a finite set of states.
In principle, agents could trade contracts for every instant $t\in[0,24]$ or for every possible realization $s\in\mathbb{R}$. In practice, however, trading at such a high level of granularity would entail substantial transaction costs, which are likely to outweigh the associated economic benefits.}
This necessitates criteria for choosing among them. We develop such criteria and show that any collection of states satisfying them corresponds to a solution of an optimal partitioning problem over the sample space of an underlying random variable.
To characterize these solutions, we draw on a central result from the optimal quantization literature~\citep[e.g.,][]{Graf2000,Pflug2014}, which shows that solutions to this partitioning problem correspond to a centroidal Voronoi partition.
This approach yields states that are both intuitive and easy to interpret, and that can be computed by solving a location problem that optimally places $S$ points in $k$-dimensional space~\citep{Eiselt2011location}.

The remainder of the article is structured as follows. In \Cref{sec: theory}, we apply the theory of equilibrium under uncertainty, as for example presented in Chapter 19 of \cite{mas1995microeconomic}, to the day-ahead electricity market. In \Cref{sec: auction design}, we discuss the implications for day-ahead auction design and incentives. In \Cref{sec: state definition}, we examine how states of the world can be defined, and in \Cref{sec: case study}, we illustrate their computation and interpretation in a case study of offshore wind farms in the European North Sea. Finally, \Cref{sec: conclusion} concludes and outlines directions for implementation of the proposed auction design.

\section{Equilibrium Under Uncertainty}\label{sec: theory}

In \Cref{subsec: model}, we introduce a general model of the power system that captures the interaction between day-ahead decisions and uncertainty, and we define the state-contingent contracts that firms can trade.
In \Cref{subsec: equilibrium}, we apply the concept of equilibrium under uncertainty to analyze how these contracts are traded and how prices emerge, showing that trade leads to welfare maximization and socially optimal day-ahead decisions.
Finally, in \Cref{app: example price formation} we provide an example.

\subsection{Model}\label{subsec: model}

Consider a power system with a set of nodes $\mathcal{N}=\{1,\ldots,N\}$ and a set of agents $\mathcal{I}=\{1,\ldots,I\}$, each of whom owns or operates a component of the system. For example, an agent $i\in\mathcal{I}$ may represent a transmission system operator, a generation company, or a consumer. Let $t\in\mathcal{T}=\{1,\ldots,T\}$ denote discrete time periods (e.g., 15-minute or hourly intervals).

\begin{figure}[tb]
\centering
\begin{tikzpicture}[x=3.2cm,y=1cm,>=stealth]

    \draw[thick,->] (0,0) -- (4.4,0) node[right] {};

    \draw[thick] (0,0.12) -- (0,-0.12);
    \draw[thick] (1.0,0.12) -- (1.0,-0.12);
    \draw[thick] (2.2,0.12) -- (2.2,-0.12);
    \draw[thick] (3.0,0.12) -- (3.0,-0.12);
    \draw[thick] (4.2,0.12) -- (4.2,-0.12);

    \node[below=4pt] at (0,0) {$t=0$};
    \node[below=4pt] at (1.0,0) {$t=t_0$};
    \node[below=4pt] at (2.2,0) {$t=t_1$};
    \node[below=4pt] at (3.0,0) {$t=1$};
    \node[below=4pt] at (4.2,0) {$t=T$};

    \node[align=center,anchor=south] at (0,0.35)
        {Trade contracts \\$x_i$ and pay $\lambda x_i$};

    \node[align=center,anchor=south] at (1.0,0.35)
        {Choice of \\ decisions $z_i$};

    \node[align=center,anchor=south] at (2.2,0.35)
        {Uncertainty $\xi$\\ and state $s$ \\is realized};

    \node[align=center,anchor=south] at (3.6,0.35)
        {Electricity delivery $x_{is}$ \\ given state $s$};

\end{tikzpicture}
\caption{Timeline of the model.}
\label{fig:timing_model}
\end{figure}
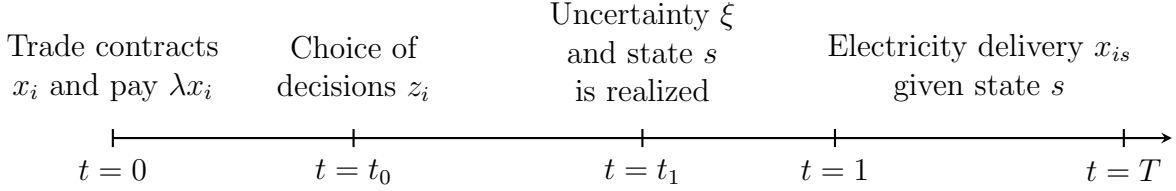

Let $z_i\in\mathbb{R}^{m_i}$ denote the decisions that agent $i$ must take at time $t=t_0$, before operations begin. These may include, for example, start-up decisions, network configurations, or production planning choices, and may involve both continuous and discrete variables.

At the time of decision-making, some parameters $\xi\in\mathbb{R}^k$ are uncertain but become realized at time $t=t_1$, before injections and withdrawals occur at $t=1$. For instance, $\xi$ may represent wind speeds at different nodes.

Before making decisions $z_i$, agents can trade contracts at $t=0$ for electricity delivery in periods $t\in\mathcal{T}$. Crucially, delivery is contingent on the realization of a state $s$: agents are required to inject or withdraw power only if state $s$ occurs. If state $s$ does not occur, no delivery takes place. 

We formalize the notion of a state as follows.

\begin{definition}[State]\label{definition: state}
Let $\xi$ be a random variable at $t=0$ with sample space $\Omega \subseteq \mathbb{R}^k$, whose realization becomes publicly known at time $t=t_1$.
A state $s$ is a non-empty subset $\Omega_s \subseteq \Omega$.
We say that state $s$ occurs if $\xi \in \Omega_s$ at time $t=t_1$.
\end{definition}

For example, let $\xi$ denote the wind speed (in m/s) at time $t=t_1$ at a given node.
If wind speed is bounded by $10$ m/s, then $\Omega=[0,10]$. A state $s$ could be defined as a singleton, e.g.\ $\Omega_s=\{1.32\}$, or as an interval, e.g.\ $\Omega_s=[1.7,3.1]$.

We formalize a state-contingent electricity contract as follows.

\begin{definition}[State-Contingent Electricity Contracts]\label{def: contracts}
Let $x_{nts}$ denote a right to withdraw (if $x_{nts}>0$) or an obligation to inject (if $x_{nts}<0$) one MW of electric power at node $n\in\mathcal{N}$ during period $t\in\mathcal{T}$, with delivery taking place if and only if state $s\in\mathcal{S}$ occurs.
Collecting all such contracts, we define
\begin{equation*}
x = (x_{nts})_{n\in\mathcal{N}, t\in\mathcal{T}, s\in\mathcal{S}} \in \mathbb{R}^{N \cdot T \cdot S}
\end{equation*}
as the vector of contracts across all nodes, periods, and states.
To acquire a portfolio $x$ at $t=0$, an agent pays (or receives)
\begin{equation*}
\lambda \cdot x = \sum_{n\in\mathcal{N}} \sum_{t\in\mathcal{T}} \sum_{s\in\mathcal{S}} \lambda_{nts} \cdot x_{nts},
\end{equation*}
where
\begin{equation*}
\lambda = (\lambda_{nts})_{n\in\mathcal{N}, t\in\mathcal{T}, s\in\mathcal{S}} \in \mathbb{R}^{N \cdot T \cdot S}
\end{equation*}
denotes the vector of contract prices.
\end{definition}

Contracts currently traded in European an US day-ahead auctions are a special case of those contracts with $S=1$ given by $\Omega_1=\Omega$.\endnote{In the current single day-ahead auction in Europe (with auction time $t=0=\text{11 am}$), contracts can be traded for $T=96$ time periods (15-minute intervals of the next day) and $N=61$ locations (bidding zones across 27 countries)~\citep[cf.][]{nemo-committee}. These contracts are not state-contingent; i.e., there is a single state $S=1$ with $\Omega_1=\{\Omega\}$.}

A key feature of these state-contingent contracts is that payments are made at $t=0$, while delivery occurs at $t\in\mathcal{T}$. In particular, although delivery is contingent on the realization of a state, payments are not~\citep{mas1995microeconomic}.

\begin{example}[State-Contingent Contracts]\label{example: contracts}
Consider a setting with $T=1$, $N=1$, and $S=2$, where the states are $\Omega_1=\text{``high wind''}$ and $\Omega_2=\text{``low wind''}$. 
Suppose the corresponding prices are $\lambda_1=10$ €/MWh and $\lambda_2=20$ €/MWh.
A wind farm with expected production of 10 MWh in state 1 and 5 MWh in state 2 can sell $x_1=10$ and $x_2=5$ units of the contracts, yielding total revenues of $10\cdot 10 + 20\cdot 5 = 200$ €.
An inflexible load of 10 MW can purchase 10 units of each contract, resulting in total payments of $10\cdot 10 + 20\cdot 10 = 300$ €.
If instead there were only a single state, $\Omega_1=\text{``high and low wind''}$, the wind farm would face the problem described in the introduction. By distinguishing between two states, this issue is resolved.
\end{example}

If state $s$ is realized, i.e., $\xi\in\Omega_s$, and an agent $i$ has bought or sold
\[
x_i = (x_{ints})_{n\in\mathcal{N},\, t\in\mathcal{T},\, s\in\mathcal{S}},
\]
then it must inject or withdraw $x_{is}$ units of electricity in that state. The resulting utility (or cost) for agent $i$ is represented by
\begin{equation*}\label{eq: utility}
u_i(z_i,x_{is},\xi): \mathbb{R}^{m_i} \times \mathbb{R}^{N\times T} \times \mathbb{R}^k \to \mathbb{R} \cup \{-\infty\},
\end{equation*}
where $u_i(z_i,x_{is},\xi)=-\infty$ indicates that the combination $(z_i,x_{is},\xi)$ of decisions, injections or withdrawals, and parameters is infeasible for agent~$i$.

\begin{remark}[Transmission System Operator]
This framework is sufficiently general to represent all agents in the power system. Generation units, storage devices, and flexible demand can all be captured through appropriate utility functions $u_i(\cdot)$. In particular, the transmission system operator (TSO) can be represented by a utility function whose domain encodes network constraints such as power flow equations and line capacity limits.
Thus, the TSO effectively buys electricity at some nodes and sells it at others, subject to the transmission grid's capacity constraints, thereby determining the extent to which state-contingent flows between nodes are permitted (maximizing congestion rent).
The TSO may trade the same quantity of contracts across all states $s$, thereby fixing flows already at $t=0$, or trade different quantities across states, allowing flows to adjust once more information about $\xi$ becomes available.
\end{remark}

\begin{remark}[Reserve Procurement]
The model can be extended to include the simultaneous trading of energy and reserves (``co-optimization''). Reserves constitute a distinct commodity that is procured by the TSO and supplied by market participants~\citep{papavasiliou2024optimization}. The TSO may procure reserve quantities for each time period $t$, location $n$, and state $s$, either uniformly across states or state-dependently. Modeling reserves as state-contingent commodities also enables variable renewable generators, such as wind and solar power plants, to provide downward reserves (curtailing production) in states with high wind or solar output.
\end{remark}

\begin{remark}[Two-Stage vs.\ Multi-Stage]
We restrict attention to uncertainty that is resolved before $t=1$, yielding a two-stage setting: first-stage decisions $z_i$ and contract positions $x_i$ are chosen at $t=0$, while consumption and production $x_{is}$ for all $t\in\mathcal{T}$ occur in the second stage after uncertainty $\xi$ is realized.  
Our analysis can be extended, without conceptual difficulty, to multi-stage settings in which uncertainty is revealed sequentially (e.g., at $t=t_2$, $t=t_3$, etc.).
\end{remark}

\subsection{Equilibrium}\label{subsec: equilibrium}

To analyze how agents trade contracts at $t=0$, we begin with the following simplifying assumption.

\begin{assumption}[Complete Contracts]\label{ass: complete market}
The sample space is finite, $\Omega=\{\xi_1,\ldots,\xi_S\}$ with $S<\infty$, and contracts are defined for each state $\Omega_s=\{\xi_s\}$, $s\in\mathcal{S}=\{1,\ldots,S\}$.
\end{assumption}

Under this assumption, no trade takes place after $t=0$ in intraday or balancing markets, since agents can fully specify their electricity consumption and production at $t=0$ for every time, location, and state~\citep{mas1995microeconomic}. This allows us to ignore subsequent markets and focus exclusively on trade at $t=0$. We discuss this assumption in detail in \Cref{subsec: incomplete day-ahead auction}.

We make a second simplifying assumption, which we discuss later in \Cref{subsec: auction} when introducing the auction design.

\begin{assumption}[Price-Taking]\label{ass: price taking}
Each agent $i\in\mathcal{I}$ assumes that its own trading decisions $x_i$ do not affect contract prices $\lambda$.
\end{assumption}

While the cost of trading a portfolio $x_i$ is certain and given by $\lambda \cdot x_i$, the utility from injecting or withdrawing electricity depends on the realization of $\xi$ and thus on the state $s$. Hence, agent $i$ must hold subjective beliefs about $\xi$, represented by probabilities $\pi_{is}$ assigned to the possible outcomes $\xi_s$ under \Cref{ass: complete market}.

If the agent maximizes expected utility, its problem is
\begin{align*}
\max_{z_i,x_{is}} \; \sum_{s\in\mathcal{S}} \pi_{is} \, u_i(z_i,x_{is},\xi_s) - \lambda \cdot x_i.
\end{align*}
More generally, the agent may evaluate uncertain outcomes using a risk measure $F_i(\cdot)$. Its problem then becomes
\begin{align*}
\max_{z_i,x_{is}} \; F_i\Big( \big(u_i(z_i, x_{is}, \xi_s)\big)_{s \in \mathcal{S}} \, ,\; (\pi_{is})_{s \in \mathcal{S}} \Big) \ - \ \lambda \cdot x_i .
\end{align*}
For example, $F_i(\cdot)$ may represent expected utility as described above, conditional value-at-risk, or the worst-case value. Note that payments $\lambda \cdot x_i$ are certain and therefore are not part of the risk evaluation captured by $F_i(\cdot)$.

\begin{example}[Trading Contracts]\label{example: trading contracts}
Consider a thermal power plant in the setting of \Cref{example: contracts}. The plant must decide at $t=0$ whether to be online at $t=1$. If it is online, it can produce between 10 and 20 MWh at a cost of 30 €/MWh. The plant believes both states are equally likely and maximizes expected profit.
Its optimization problem is thus:
\begin{align*} 
\max_{z,x} & \quad 0.5 \cdot (-30) \cdot (x_1 + x_2) + 10 \cdot x_1 + 20 \cdot x_2 \\ 
\text{s.t.} & \quad 10 \cdot z \le x_s \le 20 \cdot z \quad s=1,2 , \\
& \quad z \in \{0,1\} . 
\end{align*}
The optimal solution is $z^\ast=1$, $x_1^\ast=10$, and $x_2^\ast=20$. The plant thus sells 10 MWh under state-1 contracts and 20 MWh under state-2 contracts. 
Since payments are not state-contingent, the plant receives an upfront payment of $10\cdot10 + 20\cdot20 = 500$\,€. By contrast, delivery is state-contingent, and therefore so are production costs and total profit. If state 1 realizes, the plant must produce 10 MWh at a cost of 300\,€, yielding a profit of 200\,€. If state 2 realizes, it must produce 20 MWh at a cost of 600\,€, resulting in a loss of 100\,€. The expected profit is therefore 50\,€.

If the plant instead seeks to avoid losses at all cost, it may maximize the worst-case profit:
\begin{align*} \max_{z,x} & \quad \min \big\{-30 \cdot x_1, \; -30 \cdot x_2 \big\} + 10 \cdot x_1 + 20 \cdot x_2 \\ 
\text{s.t.} & \quad 10 \cdot z \le x_s \le 20 \cdot z \quad s=1,2 , 
\\ & \quad z \in \{0,1\} . 
\end{align*}
In this case, an optimal solution is $z^\ast=1$ and $x_1^\ast = x_2^\ast = 10$, or alternatively $z^\ast=0$ and $x_1^\ast = x_2^\ast = 0$, both yielding a worst-case profit of 0 €.
\end{example}

If there are prices $\lambda^\ast$ so that supply and demand for all contracts balance, we refer to them as (Walrasian) equilibrium prices.

\begin{definition}[Walrasian Equilibrium]\label{def: equilibrium}
An allocation $x^\ast=(x_i^\ast)_{i\in\mathcal{I}}$ and a price vector $\lambda^\ast$ constitute a Walrasian equilibrium if
\begin{align}\label{eq: agents maximisation problems}
(z_i^\ast,x_i^\ast) \in \; \argmax_{z_i,x_{is}} \; F_i\Big( \big(u_i(z_i, x_{is}, \xi_s)\big)_{s \in \mathcal{S}} \, ,\; (\pi_{is})_{s \in \mathcal{S}} \Big) - \lambda^\ast \cdot x_{i} 
\end{align}
and $\sum_{i\in\mathcal{I}} x_i^\ast = 0$. We refer to a vector $z^\ast=(z_i^\ast)_{i\in\mathcal{I}}$ as individual optimal decisions associated with equilibrium $(x^\ast,\lambda^\ast)$.
\end{definition}

We are now in a position to show that trading state-contingent contracts resolves the problem outlined in the introduction. In particular, it ensures an efficient use of renewable energy -- i.e., the allocation $x_s^\ast$ maximizes welfare across all possible states $s$ -- and yields optimal advance decisions for the system -- i.e., $z^\ast$ solves the stochastic welfare maximization problem.  
In other words, $(x^\ast,z^\ast)$ coincide with the allocation and decisions chosen by a central planner maximizing social welfare. This result is a special case of the welfare theorems~\citep{mas1995microeconomic}.

\begin{theorem}\label{theorem: equilibrium}
Let $(x^\ast,\lambda^\ast)$ be a Walrasian equilibrium and $z^\ast$ be associated individual optimal decisions. Then $(z^\ast,x^\ast)$ solves
\begin{subequations}\label{eq: central stochastic optimization}
\begin{align}
    \max_{z_i,x_{is}} \quad & \sum_{i\in\mathcal{I}} F_i\Big( \big(u_i(z_i, x_{is}, \xi_s)\big)_{s \in \mathcal{S}} \, ,\; (\pi_{is})_{s \in \mathcal{S}} \Big) \\
    \text{s.t.} \quad & \sum_{i\in\mathcal{I}} x_{i} = 0 .
\end{align}
\end{subequations}
\end{theorem}
\begin{proof}
Since $(z^\ast,x^\ast)$ solves the individual optimization problems~\eqref{eq: agents maximisation problems}, it also solves the aggregate:
\begin{align*}
\max_{z,x} \quad
\sum_{i\in\mathcal{I}} F_i\Big( \big(u_i(z_i, x_{is}, \xi_s)\big)_{s \in \mathcal{S}} \, ,\; (\pi_{is})_{s \in \mathcal{S}} \Big) - \sum_{i\in\mathcal{I}} \lambda^\ast \cdot x_{i}.
\end{align*}
Because $x^\ast$ is an equilibrium allocation it holds that $\sum_{i\in\mathcal{I}} x_{i}^\ast = 0$. Hence, $(z^\ast,x^\ast)$ must solve
\begin{align*}
 \max_{z,x} \quad &
 \sum_{i\in\mathcal{I}} F_i\Big( \big(u_i(z_i, x_{is}, \xi_s)\big)_{s \in \mathcal{S}} \, ,\; (\pi_{is})_{s \in \mathcal{S}} \Big)
 - \sum_{s\in\mathcal{S}} \lambda^\ast_s \cdot x_{is} \\
 \text{s.t.} \quad &
 \sum_{i\in\mathcal{I}} x_{i} = 0.
\end{align*}
Finally, $ \sum_{i\in\mathcal{I}} x_{i} = 0$ implies $\sum_{i\in\mathcal{I}} \sum_{s\in\mathcal{S}} \lambda^\ast_s \cdot x_{is} = 0$.
\end{proof}

\begin{remark}[Expected Social Welfare]
In the special case in which all agents maximize expected utility and share a common probability assessment, i.e., $\pi_{is}=\pi_s$ for all $i\in\mathcal{I}$ and $s\in\mathcal{S}$, the welfare maximization problem~\eqref{eq: central stochastic optimization} becomes
\begin{subequations}
\begin{align*}\label{eq: expected social welfare}
    \max_{z_i,x_{is}} \quad 
    & \sum_{s\in\mathcal{S}} \pi_s \sum_{i\in\mathcal{I}} u_i(z_i,x_{is},\xi_s) \\
    \text{s.t.} \quad 
    & \sum_{i\in\mathcal{I}} x_{is} = 0 \quad \forall s\in\mathcal{S}.
\end{align*}
\end{subequations}
\end{remark}

Equilibrium prices $\lambda_s^\ast$ reflect three factors: (i) agents’ probability assessments $(\pi_{is})_{s \in \mathcal{S}}$, (ii) their risk preferences $F_i(\cdot)$, and (iii) their utility functions $u_i(\cdot)$.  
If a state is unlikely, agents are typically unwilling to pay much for contracts that deliver in that state. However, if electricity is scarce in that state,  as reflected in the utility functions $u_i(\cdot)$, prices will be higher.  
Moreover, risk-averse agents may still be willing to pay substantial amounts for contracts in unlikely states if those states are associated with high scarcity.  
An illustrative example of this price formation is provided next.

\subsection{Illustration of Price Formation}\label{app: example price formation}

Let $T=1$, $N=1$, and $S=2$, with states $\Omega_1=\text{``high wind''}$ and $\Omega_2=\text{``low wind''}$.  
There are three agents in the system: First, a wind farm that can produce 10 MWh in state 1 and 5 MWh in state 2 at zero cost. Second, a load that demands up to 11 MWh in each state and is willing to pay 100 €/MWh. Third, a generator that can produce between 0 and 5 MWh at a cost of 50 €/MWh, but must decide its production level in advance.

All agents maximize expected utility and share common probability assessments $\pi_1$ and $\pi_2$. Production is represented by negative quantities and consumption by positive quantities. The expected welfare maximization problem~\eqref{eq: central stochastic optimization} is thus:
\begin{subequations}
\begin{align*}
    \max_{z_i,x_{is}} \quad 
    & \sum_{s=1}^2 \pi_s \cdot (100 x_{2s} + 50 x_{3s}) \\
    \text{s.t.} \quad 
    & x_{1s} + x_{2s} + x_{3s} = 0 \quad s=1,2, \\
    & -10 \le x_{11} \le 0, \\
    & -5 \le x_{12} \le 0, \\
    & 0 \le x_{2s} \le 11 \quad s=1,2, \\
    & x_{3s} = z_3 \quad s=1,2, \\
    & -5 \le z_3 \le 0.
\end{align*}
\end{subequations}

It is straightforward to verify that the duality gap of this welfare maximization problem is zero, since convexity and Slater’s condition are satisfied. It follows that an equilibrium exists in this market~\citep{mas1995microeconomic,hubner2025approximate}.
\Cref{tab:primal_balance_duals} reports the resulting equilibrium allocations $x^\ast$ and prices $\lambda^\ast$ for different probability values $\pi_1$ and $\pi_2$.  
If $\pi_1 < \pi_2$, it is optimal to operate the generator at full capacity (5 MWh). If $\pi_1 > \pi_2$, the generator operates at a lower level (1 MWh).

\begin{table}[tb]
\centering
\small
\begin{tabular}{rrrrrrrrr}
\hline
$\pi_1$ & $\pi_2$ & $x_{11}$ & $x_{12}$ & $x_{21}$ & $x_{22}$ & $z_3$ & $\lambda_1$ & $\lambda_2$ \\
\hline
0 & 1 & 0 & -5 & 5 & 10 & -5 & 0 & 100 \\
0.1 & 0.9 & -6 & -5 & 11 & 10 & -5 & 0 & 90 \\
0.2 & 0.8 & -6 & -5 & 11 & 10 & -5 & 0 & 80 \\
0.3 & 0.7 & -6 & -5 & 11 & 10 & -5 & 0 & 70 \\
0.4 & 0.6 & -6 & -5 & 11 & 10 & -5 & 0 & 60 \\
0.5 & 0.5 & -6 & -5 & 11 & 10 & -5 & 0 & 50 \\
0.6 & 0.4 & -10 & -5 & 11 & 6 & -1 & 10 & 40 \\
0.7 & 0.3 & -10 & -5 & 11 & 6 & -1 & 20 & 30 \\
0.8 & 0.2 & -10 & -5 & 11 & 6 & -1 & 30 & 20 \\
0.9 & 0.1 & -10 & -5 & 11 & 6 & -1 & 40 & 10 \\
1 & 0 & -10 & -5 & 11 & 6 & -1 & 50 & 0 \\
\hline
\end{tabular}
\caption{Allocations and prices for different probability assessments.}
\label{tab:primal_balance_duals}
\end{table}

The prices $\lambda_1$ and $\lambda_2$ reflect both the scarcity of electricity in each state and the likelihood of that state.  
In state 1 (high wind), electricity is abundant, resulting in low prices. As the probability of this state increases, however, prices rise because agents place greater value on contracts delivering in that state (``law of supply and demand'').  
In state 2 (low wind), electricity is scarce, leading to higher prices. As this state becomes less likely, agents are less willing to pay for delivery in that state, and prices decline accordingly.

\section{Auction Design}\label{sec: auction design}

\Cref{theorem: equilibrium} shows that if agents trade state-contingent contracts at $t=0$ and reach a Walrasian equilibrium, then socially optimal decisions $z_i^\ast$ are implemented. Moreover, an allocation $x_s^\ast$ is determined for each state $s$ that maximizes welfare given agents’ risk preferences and probability assessments.

In practice, reaching such an equilibrium requires a centralized auction, as bilateral trading is generally considered inefficient given the complex constraints of electricity systems~\citep{milgrom2017discovering}.\endnote{Australia, which relies on bilateral trading at the day-ahead stage, can be seen as an exception~\citep{anderson2007forward}.}  
Current day-ahead electricity auctions in Europe and the United States perform this task for the special case of a single state~\citep{Cramton2017}. We now outline how this design can be extended to multiple states.

In \Cref{subsec: auction}, we describe the auction design and its key properties.  
In \Cref{subsec: incomplete day-ahead auction}, we discuss the role of the complete market assumption (\Cref{ass: complete market}) and its limitations in practice.

\subsection{Auction Properties}\label{subsec: auction}

Consider an auction at time $t=0$ in which agents $i\in\mathcal{I}$ communicate their willingness to buy or sell the state-contingent contracts defined in \Cref{def: contracts} via bid functions $\hat{u}_i$.  
The auctioneer determines an allocation $x_i^\ast$ of contracts and prices $\lambda^\ast$ such that $(x^\ast,\lambda^\ast)$ constitute a Walrasian equilibrium with respect to the reported valuations $\hat{u}_i$.  
Formally, the auction operates as follows:

\vspace{0.5cm}

\noindent
\begin{minipage}{\linewidth}
\hrule
\vspace{0.3cm}

\noindent \textbf{Bids:} Each agent $i\in\mathcal{I}$ submits a function
\[
\hat{u}_i: \mathbb{R}^{S\times N\times T}\to\mathbb{R} \cup \{-\infty\}
\]
that assigns to every vector $x_i\in\mathbb{R}^{S\times N\times T}$ of state-contingent contracts a monetary value (e.g., in €), representing the amount the agent is willing to pay or receive.

\vspace{0.2cm}
\noindent \textbf{Allocation \& Prices:} Given the reported functions $\hat{u}_i$, the auctioneer determines an allocation $x^\ast=(x_1^\ast,\ldots,x_I^\ast)$ and prices $\lambda^\ast=(\lambda^\ast_{nts})_{n\in\mathcal{N},\, t\in\mathcal{T},\, s\in\mathcal{S}}$ such that
$$x_i^\ast \in \argmax_{x_{i}} \; \hat{u}_i(x_i) - \lambda^\ast x_i,$$ and $\sum_{i\in\mathcal{I}} x_i^\ast = 0$.

\vspace{0.2cm}
\hrule
\end{minipage}
\vspace{0.2cm}

Let us assume that agents report their true valuations for state-contingent contracts, given by
\begin{equation}\label{eq: true valuation}
    \hat{u}_i(x_i) = \max_{z_i} \; F_i\Big( \big(u_i(z_i, x_{is}, \xi_s)\big)_{s \in \mathcal{S}} \, ,\; (\pi_{is})_{s \in \mathcal{S}} \Big).
\end{equation}
Suppose further that an equilibrium exists. Then \Cref{theorem: equilibrium} implies that the auction maximizes social welfare and induces optimal decisions $z_i^\ast$.  

This raises several natural questions:
\begin{enumerate}[label=(\roman*)]
    \item \textit{Incentives:} Do agents have an incentive to report their true valuations?
    \item \textit{Individual Rationality:} Do agents benefit from participating in the auction?
    \item \textit{Revenue Adequacy:} Does the auctioneer break even?
    \item \textit{Equilibrium Existence:} What if an equilibrium does not exist?
    \item \textit{Bid Formats:} Can agents express their valuations in a computationally tractable way?
\end{enumerate}

Since this is a standard Walrasian auction, these questions are well understood. Moreover, the insights for $S=1$ carry over directly to the case $S>1$. Next, we review these results and discuss their implications for extending current day-ahead auctions from $S=1$ to $S>1$.

\paragraph{Incentives.}
If the auction has sufficiently many participants, agents have an incentive to report their true valuations. This property, known as \emph{strategy-proofness in the large} (cf.\ Theorem 1 in \cite{azevedo2019strategy}), can be informally stated as follows:

\begin{fact}[Strategy-proofness in the large]
If the auction has many participants, then bidding~\eqref{eq: true valuation} is approximately a dominant strategy for each agent $i\in\mathcal{I}$.
\end{fact}

In particular, agents are incentivized to truthfully report (i) their risk preferences $F_i(\cdot)$, (ii) their probability assessments $(\pi_{is})_{s \in \mathcal{S}}$, and (iii) their utility functions $u_i(\cdot)$.  
For further discussion, see \cite{azevedo2019strategy} for theory, and \cite{graf2013measuring}, \cite{graf2021market}, and \cite{adelowo2024redesigning} for applications to electricity markets.

\paragraph{Individual rationality.}
Walrasian auctions are \textit{individually rational}, meaning that agents cannot incur a loss from participating~\citep{mas1995microeconomic}.

\begin{fact}[Individual Rationality]
If agents truthfully report their valuations~\eqref{eq: true valuation}, then participation does not lead to losses: for any allocation $x^\ast$ and prices $\lambda^\ast$,
\begin{equation*}
 \max_{z_i} \; F_i\Big( \big(u_i(z_i, x_{is}^\ast, \xi_s)\big)_{s \in \mathcal{S}} \, ,\; (\pi_{is})_{s \in \mathcal{S}} \Big) \ - \ \lambda^\ast \cdot x_i^\ast \; \ge 0,
\end{equation*}
provided that $u_i(z_i,0,\xi_s)=0$ for all $z_i$ and $\xi_s$.
\end{fact}

Note that individual rationality does not imply that an agent cannot incur a loss in a particular state $s$. Rather, it is up to the agent to decide whether to accept such risk or to avoid it. For instance, an agent may choose to maximize expected value or adopt a worst-case criterion. If the agent wishes to avoid losses in every state, it can do so by choosing an appropriate objective. See \Cref{example: trading contracts} for an illustration.

\paragraph{Revenue adequacy.}
Walrasian auctions are \textit{budget balanced}, meaning that the auctioneer neither incurs a loss nor earns a profit~\citep{mas1995microeconomic}.

\begin{fact}[Budget Balance]
For the allocation $x^\ast$ and prices $\lambda^\ast$, net payments to the auctioneer are zero: $\lambda^\ast \cdot x^\ast = 0$.
\end{fact}

\paragraph{Equilibrium existence.}
Walrasian equilibria need not exist in markets with nonconvex preferences~\citep{mas1995microeconomic}, as is the case in electricity markets. In practice, auctioneers compute approximate equilibria using established methods~\citep{stevens2024some,hubner2025approximate}.  

Recent results show that many desirable properties of Walrasian auctions can be preserved approximately in large markets~\citep{milgrom2025walrasian}:

\begin{fact}[Approximating a Walrasian Auction]
In nonconvex markets, approximate Walrasian equilibria exist that retain many desirable properties, provided the market is sufficiently large.
\end{fact}

In practice, existing approximation methods used for $S=1$ can be extended to $S>1$. However, computational complexity increases with the number of states.

\paragraph{Bid formats.}
Ideally, agents could submit arbitrary valuation functions $\hat{u}_i$. In practice, however, bid formats must ensure computational tractability of the welfare maximization problem.

Day-ahead auctions in the U.S.\ use multi-part bids that allow thermal generators and storage units to express valuations of the form~\eqref{eq: true valuation} accurately for $S=1$~\citep{herrero2020evolving}. Extending these formats to $S>1$ is straightforward.  
For many other agents -- including wind and solar generators -- simpler bid formats are available and can likewise be extended to multiple states. 

Moreover, for agents whose valuations cannot be fully captured by these formats, XOR package bids can approximate them well given partial information about equilibrium prices $\lambda^\ast$~\citep{hubner2026package}:

\begin{fact}[XOR Bids]
If an agent has sufficient, though not perfect, information about equilibrium prices $\lambda^\ast$, they can approximate any valuation~\eqref{eq: true valuation} using a limited number of XOR package bids.
\end{fact}

Note that the introduction of state-contingent contracts need not increase bidding complexity.  
Agents who do not wish to trade state-contingent quantities (i.e., who want to buy or sell the same amount in every state) can continue to bid as usual by indicating that their bid applies uniformly across all states.  
In practice, this could be implemented through a simple interface feature, such as a checkbox specifying that the same quantities apply in every state.

\subsection{Incompleteness}\label{subsec: incomplete day-ahead auction}

So far, we have not discussed \Cref{ass: complete market}, which requires that a contract exists for every possible state of the world. In practice, this condition can only be approximated, leading to an incomplete market and necessitating trading after $t=0$.  
Importantly, incompleteness arises not only with respect to states $s$, but also with respect to time $t$ and location $n$.

If the day-ahead auction at $t=0$ were complete -- i.e., if contracts were available for every relevant time period $t$, location $n$, and state $s$ -- then no additional trading between $t=0$ and $t=1$ would be required (see Chapter~19 in \cite{mas1995microeconomic}).  
In reality, however, day-ahead auctions are not, and likely will never be, complete in any of these dimensions.

First, time is inherently continuous, with every instant of the next day lying in the interval $t\in[0,24]\subset\mathbb{R}$. In practice, however, contracts are traded only for discrete intervals (e.g., 1-hour, 30-minute, or 15-minute periods), rather than for every point in $[0,24]$.  
This incompleteness in the time dimension necessitates balancing markets that operate shortly before $t=1$ to reconcile supply and demand within these intervals~\citep{Cramton2017}.

Second, each point in the network, down to individual sockets, could be considered a distinct location requiring its own contract. In practice, however, contracts are defined at a coarser spatial resolution, such as transmission nodes (nodal pricing, as used in the U.S.) or bidding zones (zonal pricing, as used in Europe).  
This spatial incompleteness gives rise to additional mechanisms, such as redispatch markets, that operate after $t=0$ to ensure network feasibility~\citep {aravena2021transmission}.

From a theoretical perspective, an incomplete day-ahead auction may lead to inefficiencies and can, in extreme cases, perform worse than having no day-ahead auction at all~\citep{hart1975optimality,radner1982equilibrium,mas1995microeconomic,magill2002theory}.  
In particular, incompleteness can create problematic incentives for market participants, as illustrated by inc-dec gaming behavior arising from incompleteness along the location dimension~\citep{lete2026power}.  
Similarly, it would be worthwhile to analyze distorted incentives arising from incompleteness along the time and state dimensions, but this lies beyond the scope of this paper.

In practice, achieving a complete market in any of these three dimensions is infeasible. Nonetheless, there is broad consensus that day-ahead markets should be made as complete as possible with respect to time and location~\citep{newbery2018market}. This article takes the position that this principle should be extended to the state-of-the-world dimension: any choice of $S>1$ yields a market that is strictly more complete than the current design with $S=1$.

\section{States of the World}\label{sec: state definition}

In \Cref{definition: state}, we defined a state $s$ as a subset $\Omega_s$ of the sample space $\Omega$ of a random variable $\xi$ realized at $t=t_1$. In principle, any element of the power set of $\Omega$ can serve as a state. If $\Omega$ is finite, this yields $2^{|\Omega|}$ possible states; if $\Omega$ is infinite, there are infinitely many possible states.

In practice, however, only a small finite number of states can be used to define contracts, likely of a similar order of magnitude as the current numbers of time periods ($T=96$) and locations ($N=61$) in European day-ahead auctions. This raises the question of how to select an appropriate set of states.

Ideally, an auctioneer would choose the $S$ states that maximize welfare in the auction among all possible state definitions. However, if the auctioneer had sufficient information to solve this problem directly, there would be little reason to run an auction in the first place. In practice, the auctioneer can only approximate this benchmark.
In the following, we propose an approach that requires only information about the probability distribution of the random variable $\xi$ and yields states that are easy to describe and compute.

We proceed as follows. In \Cref{subsec: criteria}, we introduce and motivate the criteria we use to select a set of states. In \Cref{subsec: voronoi}, we show that the sets satisfying these criteria correspond to centroidal Voronoi partitions, and in \Cref{subsec: location problem}, we explain how they can be computed by solving a location problem. Finally, in \Cref{subsec: discussion}, we discuss the choice of these criteria.

\subsection{Criteria}\label{subsec: criteria}

The first criterion that we would like a collection of states $\Omega_1, \ldots, \Omega_S$ to satisfy is that it forms a cover of the sample space.

\begin{criterion}[Cover]\label{property: cover}
The collection of states $\Omega_1, \ldots, \Omega_S$ is a cover if $\bigcup_{s\in\mathcal{S}} \Omega_s = \Omega$.
\end{criterion}

If the collection of states does not form a cover, agents may need to trade outside the electricity auction. For example, let $\xi$ denote wind speed (in m/s) at time $t=t_1$ at a given node, with $\Omega=[0,10]$. Consider two states defined by $\Omega_1=[0,1]$ and $\Omega_2=[2,10]$. Then, for realizations $\xi\in(1,2)$, no contract is available within the auction, requiring additional trade outside the auction.

The second criterion we would like a collection of states to satisfy is pairwise disjointness.

\begin{criterion}[Pairwise disjointness]\label{property: disjointness}
The collection of states $\Omega_1, \ldots, \Omega_S$ is pairwise disjoint if $\Omega_s \cap \Omega_j = \emptyset$ for all $s \neq j$.
\end{criterion}

If this condition is violated, multiple states may occur simultaneously, complicating trade. For example, if $\Omega_1=[0,5]$ and $\Omega_2=[2,10]$, then for $\xi\in[2,5]$ both states occur, and agents would need to fulfill both contracts.

Together with the non-emptiness requirement in \Cref{definition: state}, the cover and disjointness criteria characterize a partition of the sample space.

\begin{definition}[Partition]\label{prop: partition}
A collection of nonempty subsets $\Omega_1, \ldots, \Omega_S$ of $\Omega$ is called a partition of $\Omega$ if it satisfies \Cref{property: cover} and \Cref{property: disjointness}.
\end{definition}

With Criteria~\ref{property: cover} and~\ref{property: disjointness}, we have narrowed the set of admissible state collections $\Omega_1, \ldots, \Omega_S$ to partitions of $\Omega$. To choose among these partitions, we now introduce a size criterion.
Intuitively, we would like states to be as small as possible: a singleton state $\Omega_s=\{\xi'\}$ is ideal, whereas the entire sample space $\Omega_s=\Omega$ is the worst possible state.

Our notion of size should capture both volume (Lebesgue measure) and spatial dispersion (variance) of a state $\Omega_s$. For example, the sets
\[
\Omega_s=[0,0.1]\cup[1.0,1.1]\cup\cdots\cup[9.0,9.1]
\]
and $\Omega_s=[0,1]$
have the same volume, yet we would regard the former as larger, since it is dispersed across the entire sample space $\Omega=[0,10]$, whereas the latter is concentrated.

Because the states must cover $\Omega$ by \Cref{property: cover}, some sets will inevitably be large. We would therefore like to assign smaller states to regions with high probability mass and allow larger states in regions that are less likely to occur.

For example, consider wind speed $\Omega=[0,10]$ partitioned into $S=3$ states. If wind speeds lie in $[0,4]$ with probability $90\%$, we would prefer the partition
\[
\Omega_1=[0,2], \quad 
\Omega_2=(2,4], \quad 
\Omega_3=(4,10]
\]
to the equal-width partition
\[
\Omega_1=[0,3.33], \quad 
\Omega_2=(3.33,6.66], \quad 
\Omega_3=(6.66,10],
\]
because it provides finer resolution in the more likely region.

Designing states in this way requires the auctioneer to have some knowledge of the probability distribution of the underlying random variable. Formally, we assume:

\begin{assumption}[Auctioneer's Information]\label{assumption: forecast}
The auctioneer is endowed with a probability measure $\mathcal{P}$ with finite moments on the measurable space $(\Omega, \mathcal{B})$, where $\mathcal{B}$ denotes the Borel $\sigma$-algebra on $\Omega$.
\end{assumption}

Given this probability measure $\mathcal{P}$, we define the size of a state $\Omega_s$ as the variance-weighted probability mass
\begin{equation}\label{eq: size metric}
    \text{Size}(\Omega_s) := \int_{\xi\in\Omega_s} \normtwo{\xi-\omega_s}^2 \,\diff \mathcal{P},
\end{equation}
where
\begin{equation*}
    \omega_s \ := \ \frac{1}{\mathcal{P}(\Omega_s)} \int_{\xi\in\Omega_s} \xi \,\diff \mathcal{P}
\end{equation*}
denotes the barycentre of $\Omega_s$ (also called the centroid or center of mass).

The metric~\eqref{eq: size metric} captures volume, likelihood, and dispersion simultaneously:
\begin{itemize}
    \item \emph{Volume}: integration over $\Omega_s$ reflects its geometric size; in particular, integrating the constant function $1$ under a uniform measure recovers its volume.
    \item \emph{Likelihood}: integration with respect to $\mathcal{P}$ weights each state by its probability mass.
    \item \emph{Dispersion}: weighting by squared distance to the barycentre $\omega_s$ penalizes spatial spread within the state.
\end{itemize}

To extend this notion of size from a single state $\Omega_s$ to a collection $\Omega_1,\ldots,\Omega_S$, we aggregate sizes across states. We adopt the $\ell_1$-norm, so that the size of the collection is the sum of the individual sizes:
\begin{equation*}
    \text{Size}(\Omega_1,\ldots,\Omega_S) := \sum_{s\in\mathcal{S}} \text{Size}(\Omega_s).
\end{equation*}
Minimizing this aggregate size leads to the following criterion:

\begin{criterion}[Minimal Size]\label{property: minimal}
The collection of states $\Omega_1,\ldots,\Omega_S$ has minimal size if it solves
\begin{subequations}\label{eq: optimization problem}
\begin{align}
    \min_{\Omega_1,\ldots,\Omega_S} 
    & \quad \sum_{s\in\mathcal{S}} \int_{\xi\in\Omega_s} \normtwo{\xi-\omega_s}^2 \,\diff \mathcal{P} \label{eq: optimization problem objective function} \\
    \text{s.t.} 
    & \quad \Omega_1,\ldots,\Omega_S \;\, \text{partition} \;\, \Omega, \label{eq: optimization problem constraint properties}\\
    & \quad \omega_s = \frac{1}{\mathcal{P}(\Omega_s)} \int_{\xi\in\Omega_s} \xi \,\diff \mathcal{P} \quad \forall s\in\mathcal{S} . \label{eq: optimization problem constraint barycentre}
\end{align}
\end{subequations}
\end{criterion}

Constraint~\eqref{eq: optimization problem constraint properties} ensures that the states form a partition and thus satisfy \Cref{property: cover} and \Cref{property: disjointness}. Constraint~\eqref{eq: optimization problem constraint barycentre} defines the barycentre used in the objective~\eqref{eq: optimization problem objective function}, which corresponds to the total size of the partition.

In the next subsection, we characterize the collections of states $\Omega_1,\ldots,\Omega_S$ that solve problem~\eqref{eq: optimization problem} and thus satisfy Criteria~\ref{property: cover}–\ref{property: minimal}.

\subsection{Voronoi Partition}\label{subsec: voronoi}

The partition problem~\eqref{eq: optimization problem} has been extensively studied in the literature on optimal quantization of probability distributions~\citep{Graf2000,Pflug2014} and location analysis~\citep{Eiselt2011location}.  
A central result in this literature is that its solutions are given by \emph{centroidal Voronoi partitions}.  
We briefly review the relevant concepts and refer to the above references for a detailed treatment.

We begin by defining a Voronoi tessellation.

\begin{definition}[Voronoi Tessellation]
Let \(\{p_1,\dots,p_n\}\subset \mathbb{R}^d\) be a finite set of points.  
The Voronoi cell associated with \(p_i\) is
\begin{equation*}\label{eq: voronoi cell}
    V_i=\{x\in \mathbb{R}^d \mid \|x-p_i\|_2 \le \|x-p_j\|_2 \text{ for all } j\neq i\}.
\end{equation*}
The collection \(\{V_1,\dots,V_n\}\) is called the Voronoi tessellation (or Voronoi diagram) generated by \(\{p_1,\dots,p_n\}\).
\end{definition}

The defining property of a Voronoi tessellation is that each point $x$ is assigned to the cell corresponding to its closest point $p_i$.  
\Cref{fig: voronoi paritions illustration} illustrates such tessellations in two dimensions.

\begin{figure}[tb] 
    \centering 
\begin{subfigure}{0.45\textwidth} 
    \centering 
    \includegraphics[width=\linewidth]{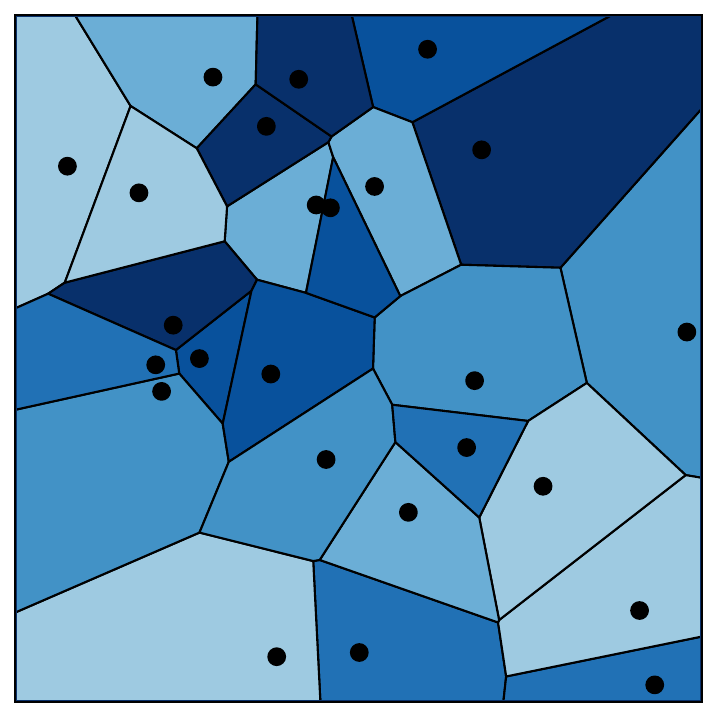} 
    \caption{Non-centroidal Voronoi Tessellation} 
\end{subfigure} 
    \hfill 
\begin{subfigure}{0.45\textwidth} 
    \centering 
    \includegraphics[width=\linewidth]{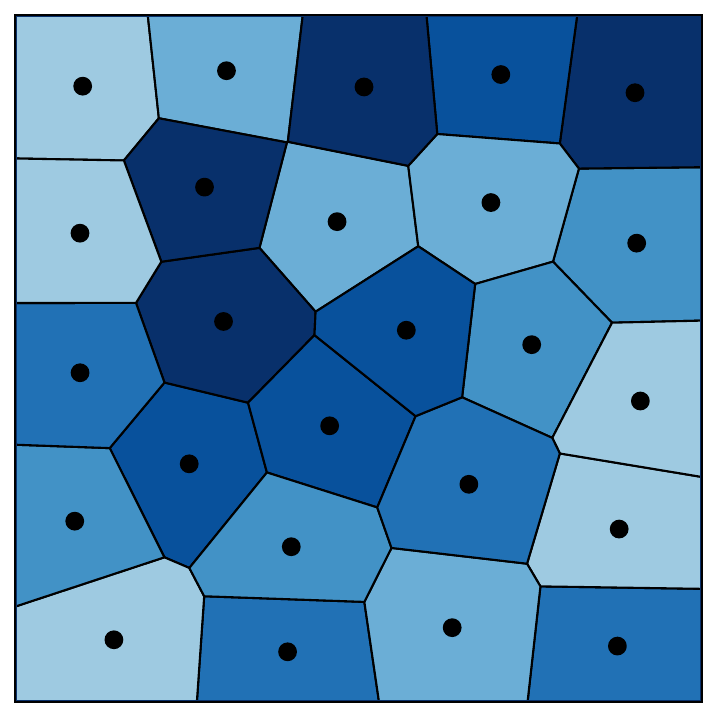} 
    \caption{Centroidal Voronoi Tessellation} 
\end{subfigure} 
    \caption{Non-centroidal and centroidal Voronoi tessellations of a two-dimensional square.} 
    \label{fig: voronoi paritions illustration} 
\end{figure}

For the cells to form a partition, boundary points must be assigned to exactly one cell. Assigning them to multiple cells would violate disjointness, while excluding them would violate the covering property.

A Voronoi tessellation is called \emph{centroidal} if each generating point $p_i$ coincides with the barycentre of its corresponding cell $V_i$ with respect to a probability measure $\mathcal{Q}$:
\begin{equation*}\label{eq: barycentre}
    \omega_i \ := \ \frac{1}{\mathcal{Q}(V_i)} \int_{x\in V_i} x \,\diff \mathcal{Q}.
\end{equation*}
In \Cref{fig: voronoi paritions illustration}, this is illustrated for the case where $\mathcal{Q}$ is the uniform distribution.

While any set of distinct points $\{p_1,\dots,p_n\}$ induces a Voronoi tessellation, it will in general not be centroidal.  
The sets of points that induce centroidal Voronoi tessellations correspond precisely to the local minimizers of the \emph{location problem}
\begin{equation*}\label{eq: location problem}
\min_{p_1,\ldots,p_n\in\mathbb{R}^d} \;\; \int_{x\in\mathbb{R}^d} \; \min_{p\in\{p_1,\dots,p_n\}} \; \normtwo{x-p}^2 \;\; \diff\mathcal{Q}.
\end{equation*}
This problem seeks to place $n$ points in $\mathbb{R}^d$ so as to minimize the average squared distance to the nearest point. It is also known as the $n$-Weber, $n$-median or continuous facility location problem~\citep{Eiselt2011location}.

We can now characterize the solutions to our partitioning problem~\eqref{eq: optimization problem}.  
Optimal partitions correspond to centroidal Voronoi partitions induced by points that are \textit{global} solutions to the location problem under the measure $\mathcal{P}$.

\begin{theorem}\label{theorem: optimal partition}
A collection of sets $\Omega_1^\ast,\ldots,\Omega_S^\ast$ solves~\eqref{eq: optimization problem} if and only if:
\begin{enumerate}[label=(\roman*)]
    \item The barycentres $\omega_1^\ast,\ldots,\omega_S^\ast$ solve
    \begin{equation}\label{eq: optimal quantization}
    \min_{\omega_1,\ldots,\omega_S\in\Omega} \;\; \int_{\xi\in\Omega} \; \min_{\omega\in\{\omega_1,\dots,\omega_S\}} \; \normtwo{\xi-\omega}^2 \;\; \diff\mathcal{P}.
    \end{equation}
    \item Each set $\Omega_s^\ast$ contains the open Voronoi cell associated with $\omega_s^\ast$:
    \begin{equation*}\label{eq: open voronoi cell}
    \Omega_s^\ast \ \supseteq \ \big\{\xi\in\Omega \mid \|\xi-\omega_s^\ast\|^2 < \min_{j\neq s} \|\xi-\omega_j^\ast\|^2 \big\}.
    \end{equation*}
    \item The sets $\Omega_1^\ast,\ldots,\Omega_S^\ast$ form a partition of $\Omega$.
\end{enumerate}
\end{theorem}

A proof can be found, for instance, in Section~4.1 of \cite{Graf2000} or Section~4.1.3 of \cite{Pflug2014}.

\Cref{theorem: optimal partition} characterizes all optimal partitions.  
To construct one, it suffices to solve~\eqref{eq: optimal quantization} and specify a rule for assigning boundary points.
A simple tie-breaking rule assigns each boundary point to the cell with the smallest index. This yields the following constructive version of \Cref{theorem: optimal partition}.

\begin{corollary}\label{corollary: constructive centroidal voronoi}
Let $\omega_1^\ast,\ldots,\omega_S^\ast$ solve~\eqref{eq: optimal quantization}. Then a solution to~\eqref{eq: optimization problem} is given by $\Omega_1^\ast, \ldots, \Omega_S^\ast$ with
\begin{equation*}
\Omega_s^\ast = \Omega_s \setminus \bigcup_{l<s} \Omega_l 
\end{equation*}
where
\begin{equation*}
\Omega_s = \big\{\xi\in\Omega \mid \|\xi-\omega_s^\ast\|^2 \le \min_{j\neq s} \|\xi-\omega_j^\ast\|^2 \big\}.
\end{equation*}
\end{corollary}

In the next subsection, we discuss how to solve problem~\eqref{eq: optimal quantization}. By \Cref{corollary: constructive centroidal voronoi}, this suffices to construct a partition that solves~\eqref{eq: optimization problem}.

\subsection{Location Problem}\label{subsec: location problem}

The optimization problem~\eqref{eq: optimal quantization} has been extensively studied in the literature on facility location~\citep{Eiselt2011location}.  
In particular, when the probability measure $\mathcal{P}$ is discrete, it can be formulated as a mixed-integer quadratic program (MIQP).

Suppose that $\mathcal{P}$ is supported on finitely many points $\{\xi_l\}_{l\in\mathcal{L}}$, where $\mathcal{L}=\{1,\ldots,L\}$, with associated probabilities $\pi_l$. Then problem~\eqref{eq: optimal quantization} can be written as
\begin{subequations}\label{eq: miqp}
\begin{align}
\min_{\omega,d,z} & \quad \sum_{l\in\mathcal{L}} \pi_l \, d_l \\
\text{s.t.} 
& \quad d_l \ge \normtwo{\xi_l - \omega_s}^2 - M(1 - z_{ls}) 
    \quad \forall l\in\mathcal{L},\; s\in\mathcal{S}, \\
& \quad \sum_{s\in\mathcal{S}} z_{ls} = 1 
    \quad \forall l\in\mathcal{L}, \\
& \quad z_{ls} \in \{0,1\} 
    \quad \forall l\in\mathcal{L},\; s\in\mathcal{S}, \\
& \quad \omega_s \in \mathbb{R}^k
    \quad \forall s\in\mathcal{S} .
\end{align}
\end{subequations}

Here, $\omega_s$ denotes the location of the $s$-th centre. The binary variable $z_{ls}$ indicates whether point $\xi_l$ is assigned to centre $\omega_s$, and $d_l$ represents the squared Euclidean distance from $\xi_l$ to its assigned centre. The parameter $M$ is a sufficiently large constant.

The constraint
\[
d_l \;\ge\; \|\xi_l - \omega_s\|_2^2 - M(1 - z_{ls})
\]
ensures that if $z_{ls}=1$, then
\[
d_l \;\ge\; \|\xi_l - \omega_s\|_2^2.
\]
Since the objective minimizes $d_l$, it follows that $d_l = \|\xi_l - \omega_s\|_2^2$ whenever point $\xi_l$ is assigned to centre $s$. Because the assignment constraints enforce that exactly one $z_{ls}$ equals one for each $l$, the variable $d_l$ represents the squared distance from $\xi_l$ to its closest centre among $\{\omega_1,\dots,\omega_S\}$.

\subsection{Discussion}\label{subsec: discussion}

While the covering and disjointness criteria (\Cref{property: cover} and \ref{property: disjointness}) are arguably uncontroversial, the minimal-size criterion (\Cref{property: minimal}) is more open to discussion.  
However, beyond ensuring that states have a small probability-weighted size, it offers two important advantages: states become easy to describe and computationally tractable.

First, the minimal-size criterion leads to a Voronoi partition. As a result, each state $\Omega_s$ can be described by a single representative point $\omega_s$ together with a notion of distance. This yields a simple and intuitive description:
\begin{quote}
    ``State $s$ occurs when $\xi$ is closer to $\omega_s$ than to any other $\omega_j$.'' 
\end{quote}
While in one dimension ($\xi\in\mathbb{R}$) this offers little advantage, already in two dimensions ($\xi\in\mathbb{R}^2$) it is significantly simpler than describing states through systems of inequalities.  
Having states that are easy to describe and understand is practically important to keep agents’ trading costs (``transaction costs'') low. Otherwise, the economic benefits of introducing additional states may be offset by the increased complexity of trading them.

Second, the minimal-size criterion yields states that can be computed by solving the MIQP~\eqref{eq: miqp}. Although this problem is NP-hard and becomes computationally challenging in high dimensions, it may remain tractable in practice, as the need for interpretability discussed above requires the dimension $k$ of the random variable $\xi\in\mathbb{R}^k$ to be small.  
For instance, in \Cref{subsec: criteria} we described states as a partition of the real line, which is easy to understand. A partition of a two-dimensional space, as illustrated in \Cref{fig: voronoi paritions illustration}, is already less transparent, and a partition of a high-dimensional space would be difficult to communicate. In practice, the number of relevant random variables is therefore likely to be small, perhaps between one and five, which helps keep the MIQP tractable.

If the problem size nevertheless becomes too large, a wide range of approximation algorithms and heuristics from the facility location literature can be applied~\citep{Eiselt2011location,Pflug2014}. While such methods may not yield exact centroidal Voronoi partitions, the resulting approximations are likely sufficient for practical purposes.

Finally, our approach assumes that the auctioneer has access to a probability measure $\mathcal{P}$ of $\xi$. In practice, much richer information is available, including detailed data on load, generation, and network conditions. This suggests that alternative criteria for defining states could be considered. However, any such approach should preserve the key advantages of the minimal-size criterion~--~namely, that states remain describable and computationally tractable.

\section{Case Study}\label{sec: case study}

In the European power system, a large share of uncertainty arises from offshore wind generation in the North Sea~\citep{schillings2012}.  
A natural way to capture this uncertainty in the day-ahead auction is to define states based on expected wind conditions in the North Sea on the following day.

In the following, we illustrate how such states can be constructed. First, in \Cref{subsec: choosing a random variable}, we discuss how to choose a random variable $\xi$ whose sample space $\Omega$ serves as the basis for defining states, and how to obtain a probability measure $\mathcal{P}$. Second, in \Cref{subsec: computing states}, we describe how to compute and interpret the resulting states.

In our case study, we focus on the computation and interpretation of states. We do not quantify the welfare gains from moving from a deterministic to a stochastic market design. This has already been studied extensively in the literature cited in the introduction, which consistently finds substantial welfare gains from such a transition.

\subsection{Random Variable}\label{subsec: choosing a random variable}

The first step is to choose a time at which wind speed is measured. As discussed in \Cref{subsec: model} and \Cref{definition: state}, the random variable must be realized before $t=1$, which corresponds to midnight in the European day-ahead auction. A natural choice is therefore the wind speed at 11 pm.

Next, the auctioneer must select the locations at which wind speed is measured. This choice should reflect both the geographical distribution of wind farms and the correlation structure of wind speeds across locations, which depends on climatic and meteorological conditions.  
A possible choice of measurement locations in the North Sea is:
\begin{itemize}
    \item Location 1: $\ang{52} \text{N}, \ang{2} \text{E}$ (off the southern English and French coast),
    \item Location 2: $\ang{54} \text{N}, \ang{7} \text{E}$ (off the Dutch, German, and Danish coast).
\end{itemize}
\Cref{fig: north sea map} shows these two measurement points. The random variable $\xi$ is then defined as
\begin{equation*}
    \xi =
    \begin{pmatrix}
        \text{Wind speed at 11 pm at Location 1} \\
        \text{Wind speed at 11 pm at Location 2}
    \end{pmatrix}.
\end{equation*}

\begin{figure}[tb]
    \centering
    \includegraphics[width=0.5\linewidth]{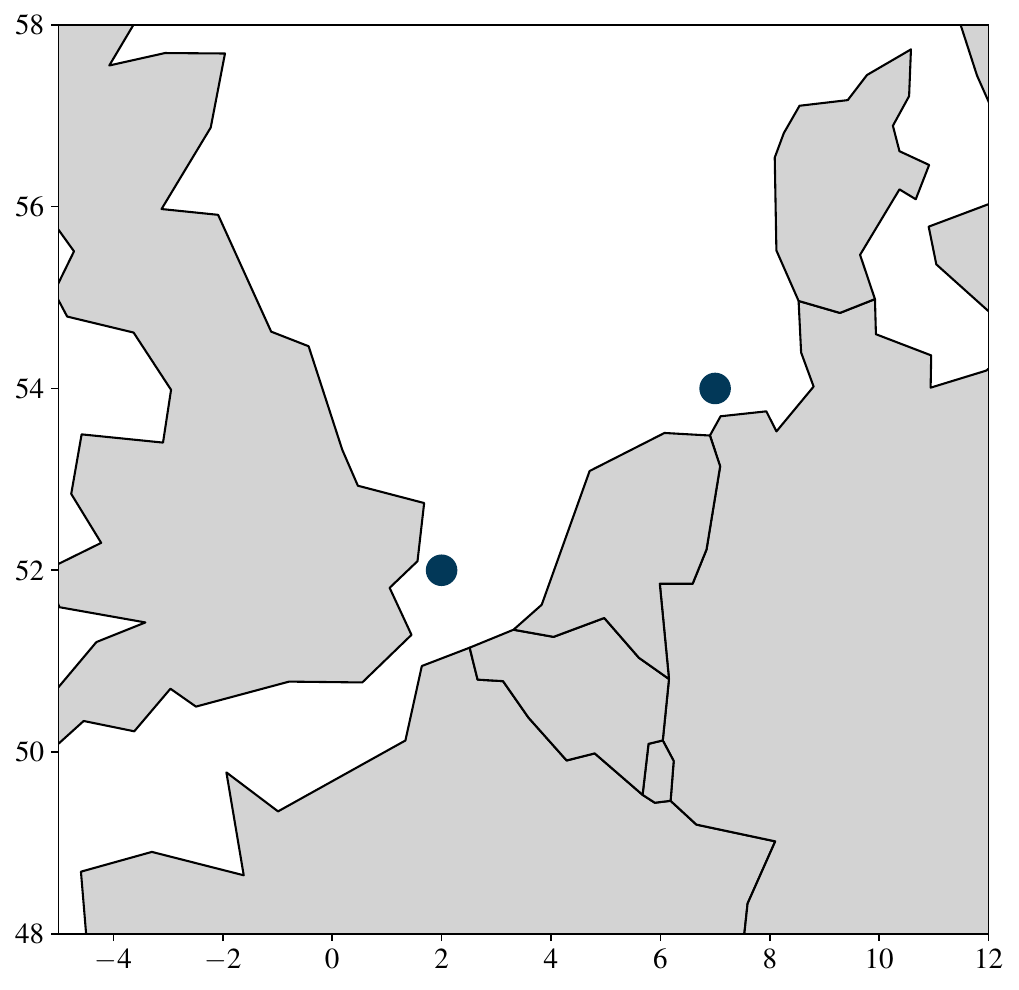}
    \caption{Locations at which wind speed is measured in the European North Sea.}
    \label{fig: north sea map}
\end{figure}

While incorporating more locations increases the completeness of contracts, it also makes the resulting states harder to interpret (cf. \Cref{subsec: discussion}). We therefore restrict attention to two locations to allow for simple graphical illustration.

\begin{remark}[Choosing a Random Variable]
The choice of the random variable is a delicate task. As illustrated in \Cref{fig: north sea map}, wind farms located close to the selected measurement points benefit from states defined by local wind conditions, whereas wind farms farther away must base their trading decisions on wind speeds measured at more distant locations.
Ideally, one would identify one or several locations whose wind speeds provide equally good predictions for all relevant wind farm sites. Conducting such an analysis for offshore wind farms in the European North Sea would be highly valuable, but lies beyond the scope of this paper.
\end{remark}

\begin{figure}[tb]
    \centering
    \includegraphics[width=0.4\linewidth]{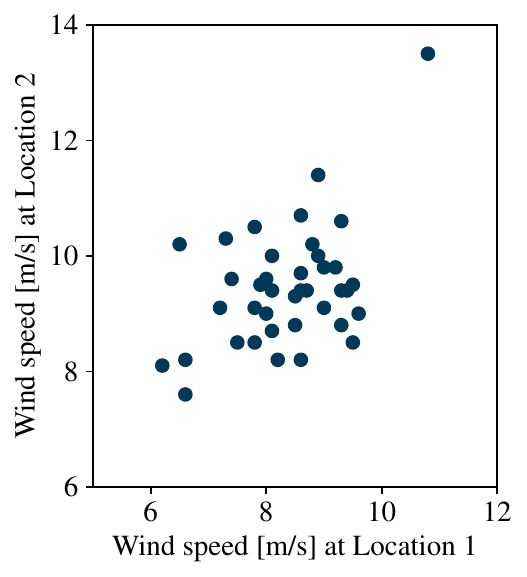}
    \caption{Distribution of wind speeds at Locations 1 and 2 at 11 pm on February 18, 2026.}
    \label{fig:forecast}
\end{figure}

Given $\xi$, the auctioneer\endnote{In the European context, the auctioneer is best represented by the committee of nominated electricity market operators~\citep{nemo-committee}, which is mandated by the European Union to operate the day-ahead and intraday markets~\citep{acer_market_rules}.} must specify how contracts are differentiated across states before the day-ahead auction begins at 11 am. Ideally, this information is announced sufficiently early to allow market participants to prepare their bids. One possible timeline is to announce the state definition at 9 am, based on forecasts generated between 8 am and 9 am.

The auctioneer can obtain a probability distribution $\mathcal{P}$ for $\xi$ from meteorological services. For illustration, we use the platform~\cite{openmeteo_ensemble_api} to access forecasts from the German National Meteorological Service (DWD) and construct a probabilistic forecast for February 18, 2026.  
The resulting distribution consists of 39 equally weighted scenarios and is shown in \Cref{fig:forecast}.

\subsection{States of Wind Speed}\label{subsec: computing states}

We now construct $S$ states $\Omega_1,\ldots,\Omega_S$ by solving the optimal partitioning problem~\eqref{eq: optimization problem}.  
In this setting, the sample space of $\xi$ is a bounded subset of $\mathbb{R}^2_{\ge 0}$, and the probability measure $\mathcal{P}$ assigns mass to 39 points.

We solve the MIQP~\eqref{eq: miqp} using Gurobi 12 for $S=2,3,$ and $4$. In all cases, the solver reaches the imposed time limit of 120 seconds, and we therefore report the best solutions found. Note that the resulting partitions are still Voronoi partitions, but not necessarily centroidal, as optimality is not guaranteed (cf. \Cref{subsec: discussion}).

\Cref{tab:branch_and_cut_summary} reports objective values and lower bounds.  
\Cref{tab: state_definitions} presents the resulting states and their interpretation:
\begin{quote}
    ``State $s$ occurs when wind speeds at Locations 1 and 2 are closer to $\omega_s$ at 11 pm than to any other $\omega_j$.'' 
\end{quote}
The corresponding partitions are illustrated in \Cref{fig: partition_wind_speed}.

\begin{table}[tb]
\centering
\small
\begin{tabular}{c c c}
\hline
$S$ & Objective value & Best lower bound \\
\hline
2 & 1.25 & 1.09 \\
3 &  0.84 & 0.3 \\
4 &  0.63 & 0.03 \\
\hline
\end{tabular}
\caption{Objective value and best lower bound found by Gurobi 12 for the MIQP~\eqref{eq: miqp} within a time limit of 120 seconds.}
\label{tab:branch_and_cut_summary}
\end{table}

\begin{table}[tb]
\centering
\small
\begin{tabular}{ccc}
\toprule
State & Interpretation & Defining point $(\omega_s)$ \\
\midrule
1 & Generally Low wind & $(7.6,8.9)$ \\
2 & Generally High wind & $(9.0,9.9)$ \\
\midrule
1 & Generally Low wind & $(7.5,9.0)$ \\
2 & Generally High wind & $(9.8,12.4)$ \\
3 & High wind at Location 1 & $(8.9,9.6)$ \\
\midrule
1 & Generally Low wind & $(6.5,8.0)$ \\
2 & Generally High wind & $(9.8,12.4)$ \\
3 & Low wind at Location 1 & $(7.7,9.3)$ \\
4 & High wind at Location 1 & $(9.0,9.5)$ \\
\bottomrule
\end{tabular}
\caption{Description of states for $S=2,3,4$. }
\label{tab: state_definitions}
\end{table}

\begin{figure}[tb]
    \centering
    \begin{subfigure}{0.32\textwidth}
        \centering
        \includegraphics[width=\linewidth]{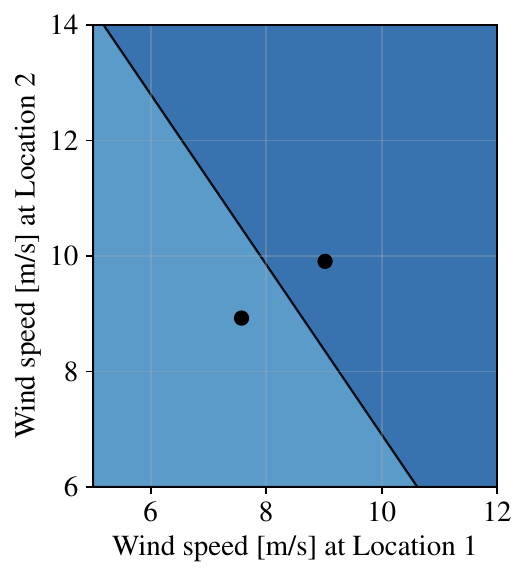}
        \caption{Two states $S=2$}
    \end{subfigure}
    \hfill
    \begin{subfigure}{0.32\textwidth}
        \centering
        \includegraphics[width=\linewidth]{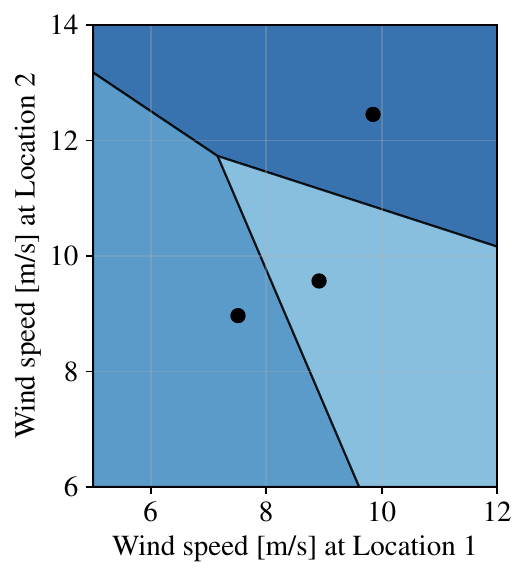}
        \caption{Three states $S=3$}
    \end{subfigure}
    \hfill
    \begin{subfigure}{0.32\textwidth}
        \centering
        \includegraphics[width=\linewidth]{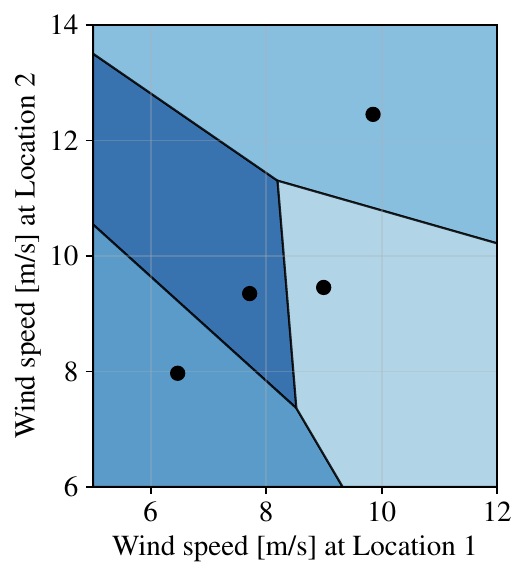}
        \caption{Four states $S=4$}
    \end{subfigure}    
    \caption{Collection of states induced by Voronoi partitions of the wind-speed plane.}
    \label{fig: partition_wind_speed}
\end{figure}

As illustrated in \Cref{fig: partition_wind_speed}, increasing the number of states refines the partition of the wind-speed space and allows agents to condition their trading decisions more precisely on expected wind conditions. In particular, moving from two to three or four states enables the market to capture spatial differences in wind patterns across locations.

While a small number of states may suffice in low-dimensional settings, the required number of states increases with the dimensionality of the underlying random variable. As additional sources of uncertainty are incorporated, more states are needed to represent them meaningfully (``curse of dimensionality'').

\section{Conclusions} \label{sec: conclusion}

We have shown that conditioning electricity contracts on the state of the world can ensure socially efficient day-ahead decisions and the effective utilization of renewable generation. 
Moreover, the required changes to the current day-ahead auction are conceptually straightforward: the existing design can be viewed as a special case of a state-contingent market with a single state.

The central challenge for implementation lies in defining an appropriate set of states. The key difficulty is that there are infinitely many possible ways to define states, while only a finite number can be implemented in practice. Any particular choice, therefore, requires a clear and well-justified rationale in order to gain acceptance among market participants. While we have proposed a principled approach based on optimal partitioning of the underlying uncertainty, substantial further research is needed.

On the theoretical side, alternative metrics for defining states could be explored that make use of richer information than considered here. In particular, electricity systems provide detailed data on load, generation, and network conditions, which could be incorporated into the state definition.  
On the practical side, simulation frameworks are needed to empirically compare different state definitions and to quantify the benefits of state-contingent contracts in terms of higher welfare, lower electricity prices, and greater integration of renewable generation.

\begingroup \parindent 0pt \parskip 0.0ex \def\enotesize{\small} \theendnotes \endgroup

\appendix 
\setcounter{section}{0}
\renewcommand{\thesection}{Appendix \Alph{section}}

\bibliography{refs} 

@article{mays2021quasi,
  title={Quasi-stochastic electricity markets},
  author={Mays, Jacob},
  journal={INFORMS Journal on Optimization},
  volume={3},
  number={4},
  pages={350--372},
  year={2021},
  publisher={INFORMS}
}

@article{kazempour2018stochastic,
  title={A stochastic market design with revenue adequacy and cost recovery by scenario: Benefits and costs},
  author={Kazempour, Jalal and Pinson, Pierre and Hobbs, Benjamin F},
  journal={IEEE Transactions on Power Systems},
  volume={33},
  number={4},
  pages={3531--3545},
  year={2018},
  publisher={IEEE}
}

@article{morales2012pricing,
  title={Pricing electricity in pools with wind producers},
  author={Morales, Juan M and Conejo, Antonio J and Liu, Kai and Zhong, Jin},
  journal={IEEE Transactions on Power Systems},
  volume={27},
  number={3},
  pages={1366--1376},
  year={2012},
  publisher={IEEE}
}

@article{pritchard2010single,
  title={A single-settlement, energy-only electric power market for unpredictable and intermittent participants},
  author={Pritchard, Geoffrey and Zakeri, Golbon and Philpott, Andrew},
  journal={Operations research},
  volume={58},
  number={4},
  pages={1210--1219},
  year={2010},
  publisher={INFORMS}
}

@article{zavala2017stochastic,
  title={A stochastic electricity market clearing formulation with consistent pricing properties},
  author={Zavala, Victor M and Kim, Kibaek and Anitescu, Mihai and Birge, John},
  journal={Operations Research},
  volume={65},
  number={3},
  pages={557--576},
  year={2017},
  publisher={INFORMS}
}

@article{hart1975optimality,
  title={On the optimality of equilibrium when the market structure is incomplete},
  author={Hart, Oliver D},
  journal={Journal of economic theory},
  volume={11},
  number={3},
  pages={418--443},
  year={1975},
  publisher={Elsevier}
}

@article{arrow1964theory,
    author = {Arrow, K. J.},
    title = {The Role of Securities in the Optimal Allocation of Risk-bearing},
    journal = {The Review of Economic Studies},
    volume = {31},
    number = {2},
    pages = {91-96},
    year = {1964}
}

@article{abada2026market,
  title={When market incompleteness is preferable to market power: Insights from power markets},
  author={Abada, Ibrahim and Ehrenmann, Andreas},
  journal={Operations Research},
  volume={74},
  number={2},
  pages={573--595},
  year={2026},
  publisher={INFORMS}
}

@article{singhal2026truthful,
  title={Truthful Production Uncertainty in Electricity Markets: A Two-Stage Mechanism},
  author={Singhal, Shobhit and Mitridati, Lesia and Romao, Licio},
  journal={arXiv preprint arXiv:2604.02455},
  year={2026}
}

@article{wong2007pricing,
  title={Pricing energy and reserves using stochastic optimization in an alternative electricity market},
  author={Wong, Steven and Fuller, J David},
  journal={IEEE Transactions on Power Systems},
  volume={22},
  number={2},
  pages={631--638},
  year={2007},
  publisher={IEEE}
}

@article{lete2026power,
  title={Power Generation Investment Under Zonal Electricity Pricing with Market-Based Re-dispatch},
  author={L{\'e}t{\'e}, Quentin and Smeers, Yves and Papavasiliou, Anthony},
  journal={The Energy Journal},
  volume={47},
  number={1},
  pages={131--164},
  year={2026},
  publisher={SAGE Publications Sage CA: Los Angeles, CA}
}

@article{schillings2012,
title = {A decision support system for assessing offshore wind energy potential in the North Sea},
journal = {Energy Policy},
volume = {49},
pages = {541-551},
year = {2012},
author = {Christoph Schillings and Thomas Wanderer and Lachlan Cameron and Jan Tjalling {van der Wal} and Jerome Jacquemin and Karina Veum}}

@BOOK{Eiselt2011location,
  title     = "Foundations of location analysis",
  editor    = "Eiselt, H A and Marianov, Vladimir",
  publisher = "Springer",
  series    = "International Series in Operations Research \& Management
               Science",
  year      =  {2011},
  address   = "New York, NY"
}

@book{debreu1959theory,
  title={Theory of value: An axiomatic analysis of economic equilibrium},
  author={Debreu, Gerard},
  year={1959},
  publisher={Yale University Press}
}

@book{Pflug2014,
  title = {Multistage Stochastic Optimization},
  journal = {Springer Series in Operations Research and Financial Engineering},
  publisher = {Springer International Publishing},
  author = {Pflug,  Georg Ch. and Pichler,  Alois},
  year = {2014}
}

@article{zakeri2019pricing,
  title={Pricing wind: a revenue adequate, cost recovering uniform price auction for electricity markets with intermittent generation},
  author={Zakeri, Golbon and Pritchard, Geoffrey and Bjorndal, Mette and Bjorndal, Endre},
  journal={INFORMS Journal on Optimization},
  volume={1},
  number={1},
  pages={35--48},
  year={2019},
  publisher={INFORMS}
}

@article{radner1982equilibrium,
  title={Equilibrium under uncertainty},
  author={Radner, Roy},
  journal={Handbook of mathematical economics},
  volume={2},
  pages={923--1006},
  year={1982},
  publisher={Elsevier}
}

@article{bjorndal2018stochastic,
  title={Stochastic electricity dispatch: A challenge for market design},
  author={Bj{\o}rndal, Endre and Bj{\o}rndal, Mette and Midthun, Kjetil and Tomasgard, Asgeir},
  journal={Energy},
  volume={150},
  pages={992--1005},
  year={2018},
  publisher={Elsevier}
}

@book{milgrom2017discovering,
  title={Discovering prices: auction design in markets with complex constraints},
  author={Milgrom, Paul},
  year={2017},
  publisher={Columbia University Press}
}

@article{graf2013measuring,
  title={Measuring competitiveness of the EPEX spot market for electricity},
  author={Graf, Christoph and Wozabal, David},
  journal={Energy Policy},
  volume={62},
  pages={948--958},
  year={2013},
  publisher={Elsevier}
}

@article{exizidis2019incentive,
  title={Incentive-compatibility in a two-stage stochastic electricity market with high wind power penetration},
  author={Exizidis, Lazaros and Kazempour, Jalal and Papakonstantinou, Athanasios and Pinson, Pierre and De Gr{\`e}ve, Zacharie and Vall{\'e}e, Fran{\c{c}}ois},
  journal={IEEE Transactions on Power Systems},
  volume={34},
  number={4},
  pages={2846--2858},
  year={2019},
  publisher={IEEE}
}

@book{mas1995microeconomic,
  title={Microeconomic theory},
  author={Mas-Colell, Andreu and Whinston, Michael Dennis and Green, Jerry R and others},
  volume={1},
  year={1995},
  publisher={Oxford university press New York}
}

@article{ratha2023moving,
  title={Moving from linear to conic markets for electricity},
  author={Ratha, Anubhav and Pinson, Pierre and Le Cadre, H{\'e}l{\`e}ne and Virag, Ana and Kazempour, Jalal},
  journal={European Journal of Operational Research},
  volume={309},
  number={2},
  pages={762--783},
  year={2023},
  publisher={Elsevier}
}

@article{mays2024sequential,
  title={Sequential pricing of electricity},
  author={Mays, Jacob},
  journal={Energy Economics},
  volume={137},
  pages={107790},
  year={2024},
  publisher={Elsevier}
}

@article{dvorkin2025regression,
  title={Regression equilibrium in electricity markets},
  author={Dvorkin, Vladimir},
  journal={IEEE Transactions on Energy Markets, Policy and Regulation},
  year={2025},
  publisher={IEEE}
}

@article{morales2014electricity,
  title={Electricity market clearing with improved scheduling of stochastic production},
  author={Morales, Juan M and Zugno, Marco and Pineda, Salvador and Pinson, Pierre},
  journal={European Journal of Operational Research},
  volume={235},
  number={3},
  pages={765--774},
  year={2014},
  publisher={Elsevier}
}

@article{hubner2025approximate,
  title={Approximate Equilibria in Nonconvex Markets: Theory and Evidence from European Electricity Auctions},
  author={H{\"u}bner, Thomas},
  journal={arXiv preprint arXiv:2503.02464},
  year={2025}
}

@article{herrero2020evolving,
  title={Evolving bidding formats and pricing schemes in USA and Europe day-ahead electricity markets},
  author={Herrero, Ignacio and Rodilla, Pablo and Batlle, Carlos},
  journal={Energies},
  volume={13},
  number={19},
  pages={5020},
  year={2020},
  publisher={MDPI}
}

@article{newbery2018market,
  title={Market design for a high-renewables European electricity system},
  author={Newbery, David and Pollitt, Michael G and Ritz, Robert A and Strielkowski, Wadim},
  journal={Renewable and Sustainable Energy Reviews},
  volume={91},
  pages={695--707},
  year={2018},
  publisher={Elsevier}
}

@book{papavasiliou2024optimization,
  title     = {Optimization Models in Electricity Markets},
  author    = {Papavasiliou, Anthony},
  year      = {2024},
  publisher = {Cambridge University Press},
  address   = {Cambridge, UK},
}

@book{magill2002theory,
  title     = {Theory of Incomplete Markets},
  author    = {Michael Magill and Martine Quinzii},
  publisher = {The MIT Press},
  year      = {2002},
  address   = {Cambridge, MA, USA}
}

@article{hubner2026package,
  title={Package bids in combinatorial electricity auctions: Selection, welfare losses, and alternatives},
  author={H{\"u}bner, Thomas and Hug, Gabriela},
  journal={Operations Research},
  volume={74},
  number={1},
  pages={56--71},
  year={2026},
  publisher={INFORMS}
}

@article{milgrom2025walrasian,
  title={A walrasian mechanism with markups for nonconvex markets},
  author={Milgrom, Paul and Watt, Mitchell},
  journal={Review of Economic Studies},
  year={2025},
  publisher={Oxford University Press UK}
}

@misc{nemo-committee,
  title        = {Single Day-Ahead Coupling (SDAC)},
  author = {{NEMO Committee}},
  note         = {Accessed: 2026-02-14},
  year         = {2026},
  url          = {https://www.nemo-committee.eu/sdac}
}

@misc{acer_market_rules,
  title        = {Market rules},
  author       = {{Agency for the Cooperation of Energy Regulators (ACER)}},
  note         = {Accessed: 2026-02-14},
  year         = {2026},
  url          = {https://www.acer.europa.eu/electricity/market-rules}
}

@misc{openmeteo_ensemble_api,
  title        = {Ensemble API},
  author       = {{Open-Meteo}},
  howpublished = {\url{https://open-meteo.com/en/docs/ensemble-api}},
  note         = {Accessed: 2026-02-14},
  year         = {2026},
  url          = {https://open-meteo.com/en/docs/ensemble-api}
}

@article{anderson2007forward,
  title={Forward contracts in electricity markets: The Australian experience},
  author={Anderson, Edward J and Hu, Xinin and Winchester, Donald},
  journal={Energy Policy},
  volume={35},
  number={5},
  pages={3089--3103},
  year={2007},
  publisher={Elsevier}
}

@article{aravena2021transmission,
  title={Transmission capacity allocation in zonal electricity markets},
  author={Aravena, Ignacio and L{\'e}t{\'e}, Quentin and Papavasiliou, Anthony and Smeers, Yves},
  journal={Operations Research},
  volume={69},
  number={4},
  pages={1240--1255},
  year={2021},
  publisher={INFORMS}
}

@article{Cramton2017,
  title = {Electricity market design},
  volume = {33},
  number = {4},
  journal = {Oxford Review of Economic Policy},
  publisher = {Oxford University Press (OUP)},
  author = {Cramton,  Peter},
  year = {2017},
  pages = {589–612}
}

@book{Graf2000,
  title = {Foundations of Quantization for Probability Distributions},
  journal = {Lecture Notes in Mathematics},
  publisher = {Springer Berlin Heidelberg},
  author = {Graf,  Siegfried and Luschgy,  Harald},
  year = {2000}
}

@article{stevens2024some,
  title={On some advantages of convex hull pricing for the European electricity auction},
  author={Stevens, Nicolas and Papavasiliou, Anthony and Smeers, Yves},
  journal={Energy Economics},
  volume={134},
  pages={107542},
  year={2024},
  publisher={Elsevier}
}

@article{graf2021market,
  title={Market power mitigation mechanisms for wholesale electricity markets: Status quo and challenges},
  author={Graf, Christoph and La Pera, Emilio and Quaglia, Federico and Wolak, Frank A},
  journal={Working Paper Stanford University},
  year={2021}
}

@article{adelowo2024redesigning,
  title={Redesigning automated market power mitigation in electricity markets},
  author={Adelowo, Jacqueline and Bohland, Moritz},
  journal={International Journal of Industrial Organization},
  volume={97},
  pages={103108},
  year={2024},
  publisher={Elsevier}
}

@article{azevedo2019strategy,
  title={Strategy-proofness in the large},
  author={Azevedo, Eduardo M and Budish, Eric},
  journal={The Review of Economic Studies},
  volume={86},
  number={1},
  pages={81--116},
  year={2019},
  publisher={Oxford University Press}
}

\end{document}